%% file: esa_clouds.tex
\documentclass{llncs}

\usepackage{algorithmic}
\usepackage{graphicx}
\usepackage{wrapfig}
\usepackage{color}
\usepackage{citesort}
\usepackage{compress}
\usepackage{times}
\usepackage{verbatim}
\usepackage{amsmath,amsfonts}
\usepackage{enumerate,url}
\usepackage{wasysym}
\usepackage{subfig}
\usepackage{tabularx,booktabs}

\graphicspath{ {figures/} }

\begin{document}

\input{macros.tex}

\let\doendproof\endproof
\renewcommand\endproof{~\hfill\qed\doendproof}

\title{On Semantic Word Cloud Representation}

\author{
  Lukas Barth$^1$ \and
  Stephen Kobourov$^2$ \and
  Sergey Pupyrev$^2$ \and
  Torsten Ueckerdt$^3$
}

\institute{
$^1$Department of Informatics, Karlsruhe Institute of Technology\\
$^2$Department of Computer Science, University of Arizona\\
$^3$Department of Mathematics, Karlsruhe Institute of Technology
}
\pagenumbering{arabic}
\pagestyle{plain}
\maketitle

\begin{abstract}
We study the problem of computing semantic-preserving word clouds in which semantically related words are close to each other. While several heuristic approaches have been described in the literature, we formalize the underlying geometric algorithm problem: Word Rectangle Adjacency Contact (WRAC). In this model each word is a rectangle with fixed dimensions, and the goal is to represent semantically related word pairs by contacts between their corresponding rectangles. We design and analyze efficient polynomial-time algorithms for variants of the WRAC problem, show that some general variants are NP-hard, and describe several approximation algorithms.
Finally, we experimentally demonstrate that our theoretically-sound algorithms outperform the early heuristics.
\end{abstract}

\section{Introduction}\label{sec:intro}

Word clouds and tag clouds are popular tools for visualizing text. The practical tool, Wordle~\cite{wordle09} took word clouds to the next level with high quality design, graphics, style and functionality. Such word
cloud visualizations provide an appealing way to summarize the content
of a webpage, a research paper, or a political speech. Often such
visualizations are used to contrast two documents; for example, word cloud visualizations of the speeches given by the candidates in
the 2008 US Presidential elections were used to draw sharp contrast between them in the popular media.

While some of the more recent word cloud visualization tools aim to
incorporate semantics in the layout, none provide any guarantees about
the quality of the layout in terms of semantics. We propose a formal model of the problem, via a simple vertex-weighted and edge-weighted
graph. The vertices in the graph are the words in the document, with
weights corresponding to their frequency (or normalized
frequency). The edges in the graph correspond to semantic relatedness,
with weights corresponding to the strength of the relation. Each
vertex must be drawn as a rectangle or box with fixed dimensions and with area determined by its weight. The goal is to ``realize'' as many edges as possible, by contacts between their corresponding rectangles; see Fig.~\ref{fig:complexity-classes}.

\subsection{Related Work}

The early word-cloud approaches did not explicitly use semantic information, such as word relatedness, in placing the words in the cloud. More recent approaches attempt to do so. Koh {\em et al.}~\cite{maniwordle} use interaction to add semantic
relationship in their ManiWordle approach. Parallel tag
clouds by Collins {\em et al.}~\cite{collins-09} are used to visualize evolution over time with the help of
parallel coordinates.
Cui {\em et al.}~\cite{Cui_2010_wordcloud} couple trend charts with
word clouds to keep semantic relationships, while visualizing evolution over time with help of force-directed
methods.
Wu {\em et al.}~\cite{wu2011semantic} introduce a method for creating semantic-preserving word
clouds based on a seam-carving image processing method and an application of bubble sets.
Hierarchically clustered document collections are visualized with self-organizing maps~\cite{HKK96}
and Voronoi treemaps~\cite{brandes12}.

Note that the semantic-preserving word cloud problem is related to classic graph layout problems, where the goal is to draw graphs so that vertex labels are readable and Euclidean distances between pairs of vertices are proportional to the underlying graph distance between them. Typically, however, vertices are treated as points and label overlap removal is a post-processing step~\cite{dwyer05,gh10}.

In {\em rectangle representations} of graphs, vertices are axis-aligned rectangles with non-intersecting interiors and edges correspond rectangles with non-zero length common boundary. Every graph that can be represented this way is planar and every triangle in such a graph is a facial triangle. These two conditions are also sufficient to guarantee a rectangle representation~\cite{ungar,thomassen1986interval,rosenstiehl1986rectilinear,buchsbaum08,fusy2009transversal}. 
Rectangle representations play an important role in VLSI layout
and floor planning.
Several interesting problems arise when the rectangles in the representation are restricted. Eppstein \textit{et al.}~\cite{eppstein2012area} consider rectangle representations which can realize any given area-requirement or perimeter-requirement on the rectangles.
In a recent survery Felsner~\cite{felsner2013rectangle} reviews many rectangulation variants, including squarings.
 N\"ollenburg \textit{et al.}~\cite{nollenburg2013edge} consider rectangle representations of edge-weighted graphs, where edge weights are proportional to the lengths of the corresponding contact.

\begin{figure}[t]
 \centering
 \includegraphics[width=.7\textwidth]{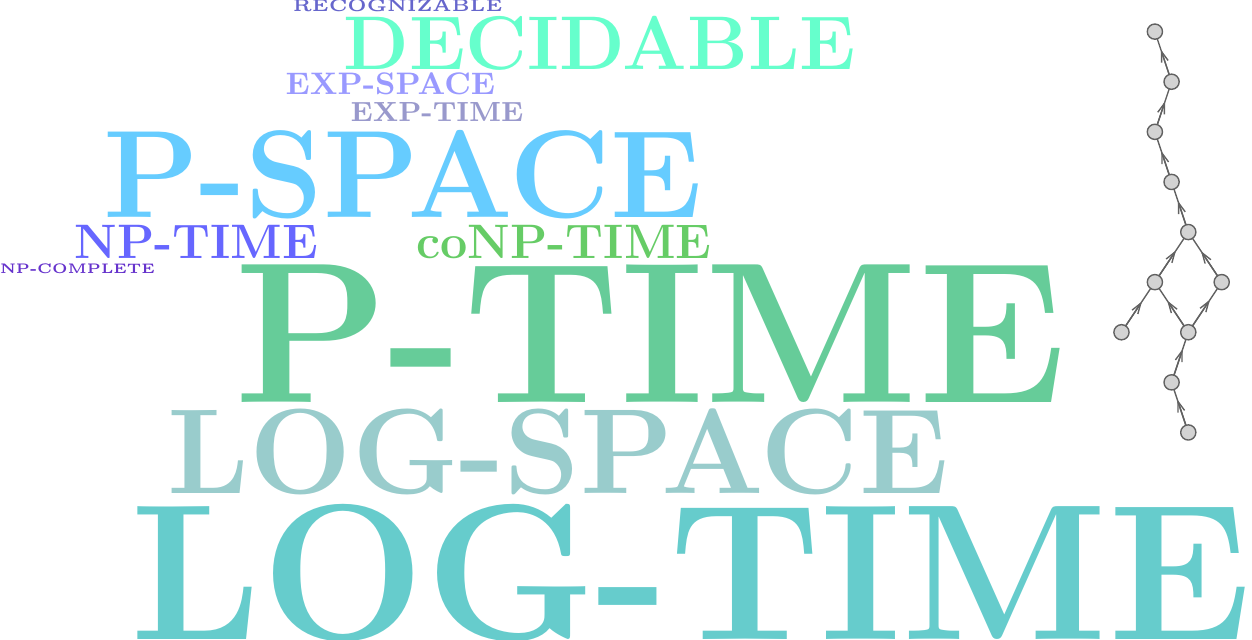}
 \caption{A hierarchical word cloud for complexity classes. A class is above another class when the first contains the second. The font size is the square root of millions of Google hits for the corresponding word. This is an example of the \emph{hierarchical \fbcr{} problem.}
}
 \label{fig:complexity-classes}
\end{figure}

\subsection{Our Contributions}

In the formal study the semantic word cloud problem we encounter several novel problems.
The input to all problems is a set of $n$ axis-aligned boxes $B_1,\ldots,B_n$ with fixed dimensions, e.g., box $B_i$ is encoded by $(w_i,h_i)$, where $w_i$ and $h_i$ its width and height.
Further, for every pair $\{i,j\}$, $i\neq j$, a non-negative \emph{profit} $p_{ij}$ represents the gain for making boxes $B_i$ and $B_j$ touch.
The set of non-zero
profits can be seen as the edge set of a graph whose vertices are the boxes, called the \emph{supporting graph}.

We define a \emph{representation of the boxes $B_1,\ldots,B_n$} to be the positions for each box in the plane, so that no two boxes overlap. 
A \emph{contact} between two boxes is a
common boundary. If two boxes are in contact, we say that these boxes \emph{touch}.
Finally, define the \emph{total profit of a representation} to be the sum of profits over all pairs of touching boxes. Next we summarize the results in this paper:

{\bf Word Rectangle Adjacency Contact (\fbcr{}):}
We are given $n$ boxes with fixed height and width each, and for each pair of boxes $B_i \neq B_j$ a profit $p_{ij}$, which is either $0$ or $1$. The task is to decide whether there exists a representation of the boxes with total profit $\sum_{i\neq j}p_{ij}$. This is equivalent to finding a representation whose induced contact graph contains the supporting graph as a subgraph. If such a representation exists, we say that it \emph{realizes the supporting graph} and that the instance of the \fbcr{} problem is \emph{realizable}.  We show that this problem is \NP{}-complete even if restricted to a tree as a supporting graph. We also show that the problem can be solved in linear time if the supporting graph is quasi-triangulated.

{\bf Hierarchical Word Rectangle Adjacency Contact (\fbcrhier{}):}
 This is a more restricted, yet useful, version of the \fbcr{} problem where the supporting graph is directed, planar, with a fixed embedding, and a unique sink. The task is to find a representation in which every contact is horizontal with the end-vertex of the corresponding directed edge on top; see Fig.~\ref{fig:complexity-classes}. We show how to solve this problem in polynomial time.

{\bf Maximum Word Rectangle Adjacency Contact (\fbcropt{}):}
 This is an optimization problem. The task is to find a representation of the given boxes, which maximizes the total profit. We show that the problem is weakly \NP{}-hard if the supporting graph is a star and  present several approximation algorithms for the problem: a constant-factor approximation for stars, trees, and planar graphs, and a $\frac{2}{\Delta+1}$-approximation for supporting graphs of maximum degree $\Delta$. We consider an extremal version of the \fbcropt{} problem and show that if the supporting graph $G=K_n$  ($n \geq 5$) and each profit is $1$, then there always exists a representation with total profit $2n-2$ and that this is sometimes best possible. Such a representation can be found in linear time.

{\bf Minimum Area Word Rectangle Adjacency Contact (\fbcrarea{}):}
 Given an instance of the \fbcr{} problem, which is already known to be realizable, find a representation that realizes the supporting graph and minimizes the area of the bounding box containing all boxes. We show that this problem is \NP{}-hard even if restricted to even simpler graphs as supporting graphs, namely independent sets, paths, or cycles.

\section{The \fbcr{} problem}\label{sec:realize}

\begin{theorem}\label{thm:trees:hardness}
  \fbcr{} is \NP{}-complete even if the supporting graph is a tree.
\end{theorem}
\begin{proof}
 It is easy to verify a solution of the \fbcr{} problem in polynomial time, so the problem is in \NP{}. To show that the problem is \NP{}-hard we use a reduction from \prob{3-Partition}, which is defined as follows. Given a multiset $S = \{s_1,s_2,\ldots,s_n\}$ of $n = 3m$ integers with $\sum_{i=1}^n s_i = mB$, is there a partition of $S$ into $m$ subsets $S_1,\ldots,S_m$ such that in each subset the numbers sum up to exactly $B$? This classical problem is known to be \NP{}-complete even if for every $i$ we have $B/4 < s_i < B/2$, in which case every subsets $S_j$ must contain exactly three elements. We also assume w.~l.~o.~g.
that $B > (m-1)/2$, which can be achieved by scaling all $s_i$ appropriately.

 Given an instance $S = \{s_1,s_2,\ldots,s_n\}$ of \prob{3-Partition}, $n = 3m$, $\sum_{i=1}^n s_i = mB$, we define a tree $T_S$ on $2n + 4$ vertices as follows. There is a vertex $v_i$ for $i=1,\ldots,n$, a vertex $w_j$ for $j=1,\ldots,m$, a vertex $u_j$ for $j=1,\ldots,m-1$, a vertex $x_j$ for $j=1,\ldots,m-1$, a vertex $c$, and five vertices $a_1,a_2,a_3,a_4,a_5$. Vertex $c$ is adjacent to all vertices except for $w_1,\ldots,w_m$ and $x_1,\ldots,x_{m-1}$. For $j = 1,\ldots,{m-1}$ vertex $u_j$ is adjacent to $w_j$ and $x_j$, and finally $u_{m-1}$ is adjacent to $w_m$; see Fig.~\ref{fig:tree:hardness}.

 \begin{figure}[t]
  \centering
  \subfloat{\includegraphics[width=.3\textwidth]{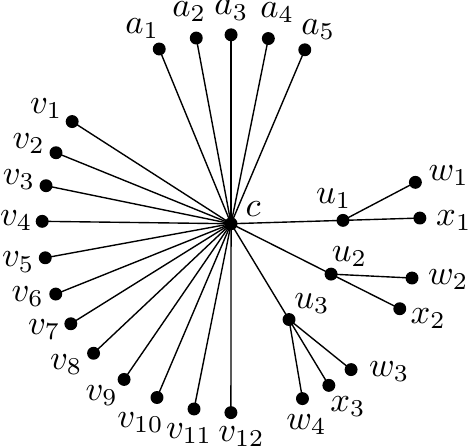}
  }
  \subfloat{\includegraphics[width=.7\textwidth]{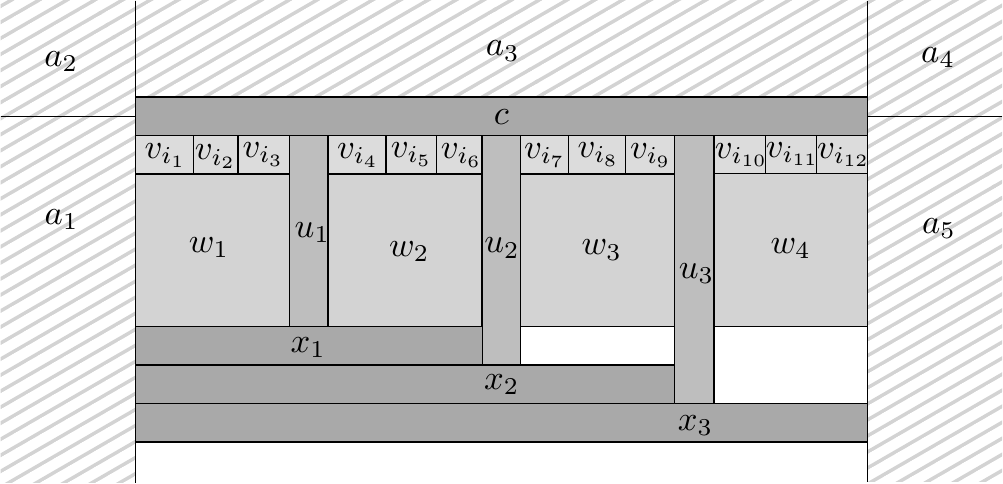}
  }
  \caption{The tree $T_S$, created from a set $S$ of $n$ integers and a representation realizing $T_S$.}
  \label{fig:tree:hardness}
 \end{figure}

 For each vertex we define a box by specifying its height and width. For simplicity let us write $v \to (h,w)$ to say that the box for $v$ has height $h$ and width $w$. Using this notation we define $u_j \to (B+j,1)$ for $j=1,\ldots,{m-1}$, $w_j \to (B,B)$ for $j=1,\ldots,m$, $x_j \to (1,iB+B+i)$ for $j=1,\ldots,m-1$, $v_i \to (1,s_i)$ for $i=1,\ldots,n$, $c \to (1,mB+m-1)$, and $a_k \to (mB+m-1,mB+m-1)$ for $k=1,2,3,4,5$.

 We claim that an instance $S$ of \prob{3-Partition} is feasible if and only if the instance of \fbcr{} defined above is feasible. To this end, consider any representation that realizes $T_S$. We refer to Fig.~\ref{fig:tree:hardness} for an illustration. We abuse notation and refer to the box for a vertex $v$ also as $v$. The box $c$ has height $1$ and width $mB+m-1$. Since $c$ touches the five $mB+m-1 \times mB+m-1$ squares $a_1$, $a_2$, $a_3$, $a_4$ and $a_5$, each $a_k$ contains a corner of $c$. It follows that at least three sides of $c$ are partially covered by some $a_k$ and at least one horizontal side of $c$ is completely covered by some $a_k$. Because $c$ has height $1$ only, but touches the boxes $v_1,\ldots,v_n,u_1,\ldots,u_{m-1}$ (each of height at least $1$), all these boxes touch $c$ on its free horizontal side, say the bottom. Indeed the widths of $v_1,\ldots,v_n,u_1,\ldots,u_{m-1}$ sum exactly to the width of $c$.

 Now $u_{m-1}$ touches $x_{m-1}$ whose width is also $mB+m-1$. Since $u_{m-1}$ has height $B+m-1 < mB+m-1$, the top of $x_{m-1}$ touches the bottom of $u_{m-1}$ and the left and right of $x_{m-1}$ touch some $a_i$ each. Since $u_{m-1}$ also touches the $B \times B$ squares $w_{m-1}$ and $w_m$ and $u_{m-1}$ has height $B+m-1 < 2B$, there is one square on each side of $u_{m-1}$. Then t $u_{m-1}$ and the rightmost $a_k$ are at horizontal distance of at least $B$.

 The height of $u_{m-2}$ is by one less than the height of $u_{m-1}$. Moreover, $u_{m-2}$ touches $x_{m-2}$ whose width is by $B + 1$ less than the width of $x_{m-1}$. This forces $x_{m-2}$ to touch some $a_k$ on the left, $u_{m-1}$ on the right and $u_{m-2}$ on top. Moreover, $u_{m-2}$ has $w_{m-2}$ on its left side. It follows that $u_{m-2}$ and $u_{m-1}$ have a horizontal distance of at least $B$.

 Similarly, for all $i=m-1,\ldots,2$ the boxes $u_i$ and $u_{i-1}$, as well as the box $u_1$ and the leftmost box $a_k$, have a horizontal distance of at least $B$. Now the width of $c$ being $mB+m-1$ forces all these distances to be \emph{exactly} $B$. Thus the boxes $v_1,\ldots,v_n$ are partitioned into $m$ subsets corresponding to the $m$ spaces between the leftmost $a_k$, all the $u_j$, and the rightmost $a_k$. Since $v_i$ has width $s_i$, $i=1\ldots,n$, in each subset the numbers sum up to exactly $B$.

 Along the same lines one can easily construct a representation realizing $T_S$ based on any given solution of the \prob{3-Partition} instance $S$.
This concludes the proof.
\end{proof}


By Theorem~\ref{thm:trees:hardness} the \fbcr{} problem is \NP{}-hard if the supporting graph is tree, and thus it is \NP{}-hard in general. However, there are classes of supporting graphs for which the problem can be solved efficiently.


A rectangle representation is called a \emph{rectangular dual} if the union of all rectangles is again a rectangle whose boundary is formed by exactly four rectangles. A graph $G$ admits a rectangular dual if and only if $G$ is planar, internally triangulated, has a quadrangular outer face and does not contain separating triangles~\cite{buchsbaum08}. Call such graphs \emph{quasi-triangulated}. The four outer vertices of a quasi-triangulated graph are denoted by $v_N$, $v_E$, $v_S$, $v_W$ 
in clockwise order around the outer quadrangle. A quasi-triangulated graph $G$ may have exponentially many rectangular duals. However, every rectangular dual of $G$ can be built up by placing one rectangle at a time, always keeping the union of placed rectangle in staircase shape.


\begin{theorem}\label{thm:quasi-triangulated}
 \fbcr{} can be solved in linear time for quasi-triangulated support graphs.
\end{theorem}
\begin{proof}[Sketch]
 The algorithm greedily builds up the quasi-planar supporting graph $G$. Start with a vertical and a horizontal ray emerging from the same point $p$, as placeholders for the right side of $v_W$ and the top side of $v_S$, respectively. Then at each step consider a \emph{concavity} -- a point on the boundary of the so far constructed representation which is a bottom-right or top-left corner of some rectangle -- with $p$ as the initial concavity. Since each concavity $p$ is contained in exactly two rectangles, there exists a unique rectangle $R_p$ that is yet to be placed and has to touch both these rectangles. If by adding $R_p$ we still have as staircase shape representation, then we do so. If no such rectangle can be added, we conclude that $G$ is not realizable. See Fig.~\ref{fig:quasi-triangulated} for an illustration; the complete proof is in the Appendix.
 \begin{figure}[t]
  \centering
  \subfloat{\includegraphics{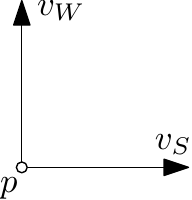}
  }\hspace{2em}
  \subfloat{\includegraphics{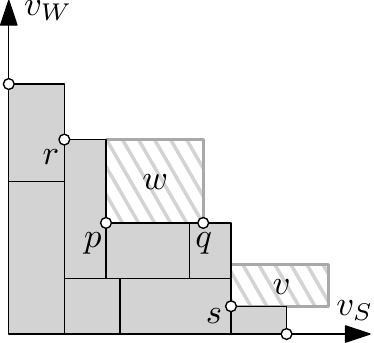}
  }
  \hspace{2em}
  \subfloat{\includegraphics{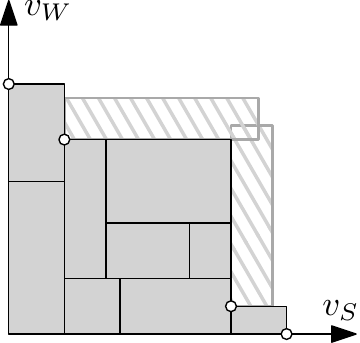}
  }
  \caption{Left: starting configuration with rays $v_S$ and $v_W$.
Center: representation at an intermediate step: vertex $w$ fits into concavity $p$ and results in a staircase, vertex $v$ fits into concavity $s$ but does not result in staircase. Adding box $w$ to the representation introduces new concavity $q$, and the vertex at concavity $r$ may be applicable. Right: there is no applicable vertex and the algorithm terminates.}
  \label{fig:quasi-triangulated}
 \end{figure}
\end{proof}

\section{The \fbcrhier{} problem}\label{sec:hierarchy}


The \fbcrhier{} problem is a more restricted variant of the \fbcr{} problem, but it can be used in practice to produce word clouds with a hierarchical structure; see Fig.~\ref{fig:complexity-classes}.
In this setting the input is a plane embedded graph $G$ with an acyclic orientation of its edges such that only one vertex has no outgoing edges, called a \emph{sink}. The task is to find a representation that \emph{hierarchically realizes $G$}, that is, it induces $G$ with its embedding as a contact graph and for every directed edge $v \to w$ in $G$ the box for $v$ touches the box for $w$ with its top side. In particular, every contact is horizontal and going along directed edges in the graph corresponds to ``going up'' in the representation.

If the embedding of $G$ is not fixed, it is easy to adapt the proof of Theorem~\ref{thm:trees:hardness} to show that the problem is again \NP{}-complete, already for trees. Indeed, one simply has to remove the vertices $a_k$, $k=1,2,3,4,5$, and orient the remaining edges of $T_S$ according to the representation shown in Fig.~\ref{fig:tree:hardness}. However, if we fix the embedding of the supporting graph $G$ and there is exactly one sink, then the \fbcrhier{} problem is polynomial-time solvable.

\begin{theorem}\label{thm:hierarchical-planar}
 The \fbcrhier{} problem can be solved in polynomial time.
\end{theorem}
\begin{proof}
 Let $G$ be the given supporting graph, i.e., a directed embedded planar graph with vertex set of boxes $\mathcal{B} = \{B_1,\ldots,B_n\}$. Let $h_i$ and $w_i$ be the height and width of box $B_i$, $i=1,\ldots,n$, and $B_1$ be the unique sink. Our algorithm consists of three phases.

{\bf Phase 1:} Here we check whether the orientation and embedding of $G$ are compatible with each other. Indeed the orientation of $G$ must be acyclic, and going clockwise around every vertex the incident edges must come as a (possibly empty) set of incoming edges followed by a (possibly empty) set of outgoing edges. If one of the two properties fails, then $G$ can not be hierarchically realized and the algorithm stops.

{\bf Phase 2:} Here we check whether the given heights of boxes are compatible with the orientation of $G$. More precisely, we set for each box $B_i$ two numbers $\mathrm{low}_i$ and $\mathrm{high}_i$, which correspond to the $y$-coordinate of the bottom and top side of $B_i$, respectively. In particular, we set $\mathrm{low}_1 = 0$, for every $i = 1,\ldots,n$ we set $\mathrm{high}_i = \mathrm{low}_i + h_i$, and for every edge $B_i \to B_j$ we set $\mathrm{high}_i = \mathrm{low}_j$. This can be done with one iteration of breadth-first search of $G$. If one number would have to be set to two different values, then $G$ can not be hierarchically realized and the algorithm stops.

{\bf Phase 3:} Here we check whether the given widths of boxes are compatible with the orientation and embedding of $G$ and compute a representation hierarchically realizing $G$, if it exists. Since we already know the $y$-coordinates for each box it suffices to compute a valid assignment of $x$-coordinates. To avoid overlaps, any two boxes whose $y$-coordinates intersect interiorly must have interiorly disjoint $x$-coordinates. Since $G$ has a unique sink we can determine which of the two boxes lies to the left and which to the right: consider for every box $B_i$ the leftmost and rightmost directed path from $B_i$ to $B_1$ and say that $B_i$ \emph{lies to the left of} $B_j$ if the leftmost path of $B_i$ joins the leftmost path of $B_j$ from the left. Similarly, $B_i$ \emph{lies to the right of} $B_j$ if the rightmost path of $B_i$ joins the rightmost path of $B_j$ from the right. Note that if $B_i$ lies to the left of $B_j$ then $B_j$ does not
lie to the left of $B_i$, but $B_i$ may also lie to the right of $B_j$. 
 More precisely, we introduce for each box $B_i$ two variables $\mathrm{left}_i$ and $\mathrm{right}_i$, which correspond to the $x$-coordinate of the left and right side of $B_i$, respectively. We consider the equations
 \begin{equation}\label{eq:hierarchical-planar-eq}
  \mathrm{right}_i = \mathrm{left}_i + w_i \hspace{5em} \text{for }i = 1,\ldots,n
 \end{equation}
 which ensure that each box $B_i$ has width $w_i$. When the $y$-coordinates of $B_i$ and $B_j$ intersect interiorly, i.e., if $\max\{\mathrm{low}_i,\mathrm{low}_j\} < \min\{\mathrm{high}_i, \mathrm{high}_j\}$, we have inequalities
 \begin{eqnarray}
  \mathrm{right}_i &\leq& \mathrm{left}_j \hspace{5.6em} \text{for $B_i$ to the left of $B_j$, and}\label{eq:hierarchical-planar-ineq1}\\
  \mathrm{left}_i &\geq& \mathrm{right}_j \hspace{5em} \text{for $B_i$ to the right of $B_j$}\label{eq:hierarchical-planar-ineq2}
 \end{eqnarray}
 which ensure that $B_i$ and $B_j$ do not intersect interiorly. Finally, for every directed edge $B_i \to B_j$ we consider the inequalities
 \begin{eqnarray}
  \mathrm{right}_i &\geq& \mathrm{left}_j \hspace{5em} \text{and} \label{eq:hierarchical-planar-edge1}\\
  \mathrm{left}_i &\leq& \mathrm{right}_j. \label{eq:hierarchical-planar-edge2}
 \end{eqnarray}
 which ensure that boxes $B_i$ and $B_j$ touch. It is easy to verify that the solutions of the system of linear equations~\eqref{eq:hierarchical-planar-eq} and inequalities~\eqref{eq:hierarchical-planar-ineq1}--\eqref{eq:hierarchical-planar-edge2} on variables $\mathrm{left}_i$ and $\mathrm{right}_i$ correspond 
to representations hierarchically realizing $G$. Thus if a solution is found, the algorithm defines a representation by placing box $B_i$ with its bottom-left corner onto the point $(\mathrm{left}_i,\mathrm{low}_i)$, $i =1,\ldots,n$. If no solution exists, then $G$ can not be hierarchically realized and the algorithm stops.

 The first two phases can be easily carried out in linear time. In the third phase, finding all leftmost and rightmost paths and deciding for every pair $B_i$, $B_j$ whether $B_i$ lies left or right of $B_j$, can also be done in linear time. Setting up the equations and inequalities takes at most quadratic time since there are $\mathcal{O}(n^2)$ inequalities.
The rest boils down to linear programming, and hence, in polynomial time. (A feasible solution can be found faster than with LP, but we leave the details out of this paper.)
\end{proof}

\section{The \fbcropt{} problem}\label{sec:optimize}
We begin by showing that \fbcropt{} is \NP{}-hard, even for simple supporting graphs. Since this version of the problem is particularly relevant in practice, we also present approximation algorithms for several different classes of supporting graphs.

\subsection{NP-hardness}

\begin{theorem}\label{thm:star-hardness}
\fbcropt{} is (weakly) \NP{}-hard if the supporting graph is a star.
\end{theorem}
\begin{proof}[Sketch]
 We use a reduction from the well-known \prob{Knapsack} problem, where the task is to decide if there exists a subset $S$ of $n$ given items, each with weight $w_i > 0$ and a profit $p_i > 0$, that fits into a knapsack with capacity $C$, i.e., $\sum_{i \in S} w_i \leq C$, and yields a total profit of at least $P$, i.e., $\sum_{i \in S} p_i \geq P$.

 The reduction is similar to the one presented in the proof of Theorem~\ref{thm:trees:hardness}. We define an edge-weighted star $S_I$ with a vertex $v_i$ for each item, a vertex $c$ which is the center of the star, and five vertices $a_1,a_2,a_3,a_4,a_5$ that block all but one side of $c$. The rectangle for each $v_i$ has width $w_i$, height $1$ and the profit for its edge with $c$ is $p_i$. The rectangle for $c$ has width $C$ and height $1$. For $k=1,2,3,4,5$ the rectangle for $a_k$ is a $C \times C$ square and the profit of the edge $a_kc$ is $\sum p_i$, which ensures that in every optimal solution these edges are realized.

 It is now straightforward to check that a subset $S$ of items can be packed into the knapsack if and only if the vertices for $S$ plus $a_1,\ldots,a_k$ can touch $c$. Details are provided in the Appendix.
\end{proof}

%

\subsection{Approximation Algorithms}
\label{sec:approx}

In this section we present approximation algorithms for the \fbcropt{} problem, for certain classes of supporting graphs. As a common tool for our algorithm we use the \prob{Maximum Generalized Assignment Problem} (GAP) defined as follows: Given a set of bins with capacity constraint and a set of items that have a possibly different size and value for each bin, pack a maximum-valued subset of items into the bins. It is known that the problem is \NP{}-complete (\prob{Knapsack} as well as \prob{Bin Packing} are special cases of \prob{GAP}), and there is a polynomial-time $(1-1/e)$-approximation algorithm~\cite{Fleischer2011}. In the remainder we assume that there is an $\alpha$-approximation algorithm for the \prob{GAP} problem, setting $\alpha = 1-1/e$.

\begin{theorem}\label{thm:approx-star}
 There exists a polynomial-time $\alpha$-approximation algorithm for the \fbcropt{} problem if the supporting graph is a star.
\end{theorem}
\begin{proof}
 Let $B_0$ denote the box corresponding to the center of the star. In any optimal solution for the \fbcropt{} problem there are four boxes $B_1,B_2,B_3,B_4$ whose sides contain one corner of $B_0$ each. Given $B_1,B_2,B_3,B_4$, the problem reduces to assigning each remaining box $B_i$ to at most one of the four sides of $B_0$ which completely contains the contact between $B_i$ and $B_0$; see Fig.~\ref{fig:approx-star}.

 \begin{figure}[t]
  \centering
  \includegraphics[width=8cm]{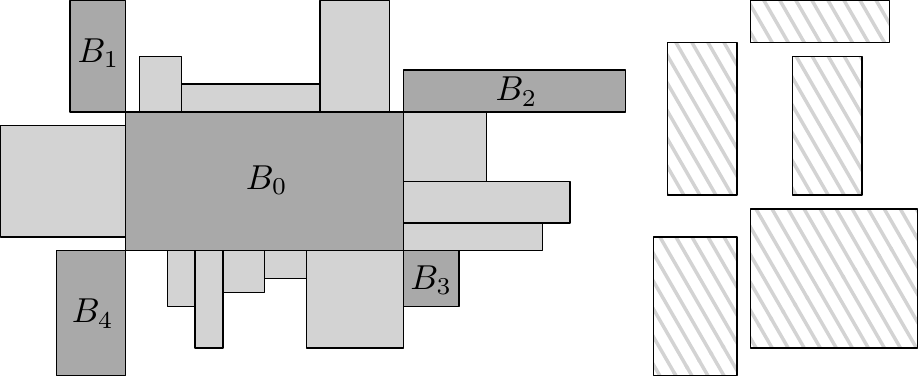}
  \caption{An optimal representation for the \fbcropt{} problem whose supporting graph is a star with center $B_0$. The striped boxes on the right are those from the trash bin.}
  \label{fig:approx-star}
 \end{figure}

 This can be formulated as the \prob{GAP} problem. The four sides are the bins plus a trash bin for all boxes not touching $B_0$, the size of an item is its width for the horizontal bins and its height for the vertical bins (the size for the trash bin is irrelevant), the value of an item is its profit of the adjacency to the central box except for the trash bin where all items have value $0$. We can now apply the algorithm for the \prob{GAP} problem, which will result in the $\alpha$-approximation for the set of boxes. To get an approximation for the \fbcropt{} problem we consider all possible variants of choosing boxes $B_1,B_2,B_3,B_4$, which increases the runtime only by a polynomial factor.
\end{proof}

A \emph{star forest} is a disjoint union of stars. A partition of a graph $G$ into $k$ star forest is a partitioning of the edges of $G$ into $k$ sets, each being a star forest.

\begin{theorem}\label{thm:approx-from-stars}
 If the supporting graph can be partitioned in polynomial time into $k$ star forests, then there exists
 a polynomial-time $\alpha/k$-approximation algorithm for the \fbcropt{} problem.
\end{theorem}
\begin{proof}
  Consider any representation with maximum total profit, that is, an optimal solution to the \fbcropt{} problem. Let $E^* \subseteq E$ be the subset of edges that are realized as contacts in this representation, and let $W_{opt}$ be the total profit of this representation. Partition the supporting graph into $k$ star forests. Since the edges of the supporting graph contain all the edges $E^*$, we find a forest $F$ with
 $$\sum_{e \in E(F)\cap E^*} p_e \geq \frac{1}{k} W_{opt}.$$
 Applying Theorem~\ref{thm:approx-star} to each star in $F$ and putting the resulting representations disjointly next to each other, gives the desired representation.
\end{proof}

\begin{corollary}\label{cor:approx}
 There is an approximation algorithm for the \fbcropt{} problem with
 \begin{itemize}
  \item approximation factor $\alpha/2$ if the supporting graph is a tree,
  \item approximation factor $\alpha/6$ if the supporting graph is planar.
 \end{itemize}
\end{corollary}
\begin{proof}
 It is easy to partition any tree into two star forests in linear time. Moreover, every planar graph can be partitioned into three trees in linear time, for example by finding a Schnyder wood~\cite{schnyderwoods}. Then the three trees can be partitioned into six star forests.
 The results now follow directly from Theorem~\ref{thm:approx-from-stars}.
\end{proof}

Our method of partitioning the supporting graph into star forests and choosing the best, is likely not optimal. Nguyen \textit{et al.}~\cite{nguyen2008approximating} show how to find a star forest carrying at least half of the profits of an optimal star forest in polynomial-time. However, we can not guarantee that the approximation of the optimal star forest carries a positive fraction of the total profit in an optimal solution of the \fbcropt{} problem. Hence, approximating the \fbcropt{} problem for general graphs remains an open problem. As a step in this direction, we present a constant-factor approximation for supporting graphs with bounded maximum degree. First we need the following lemma.

\begin{lemma}\label{lem:build-cycle}
 For every set of $n \geq 3$ boxes we can find a representation realizing any given $n$-cycle in linear-time.
\end{lemma}
\begin{proof}
 Let $C=(v_1,v_2,\ldots,v_n)$ be any given cycle. We first make boxes $v_1$ and $v_n$ adjacent horizontally; see Fig.~\ref{fig:combined}. We proceed in steps, adding one or two boxes in each step. At each step we consider the rightmost horizontal contact. Let $p$ be the rightmost point in the contact $v_i \cap v_j$. We maintain that if $v_i$ is the box on top and $v_j$ is the box below, then $i < j$ and we have placed precisely the boxes $v_k$ with $k \leq i$ or $k \ge j$.

 \begin{figure}[t]
  \centering
  \includegraphics{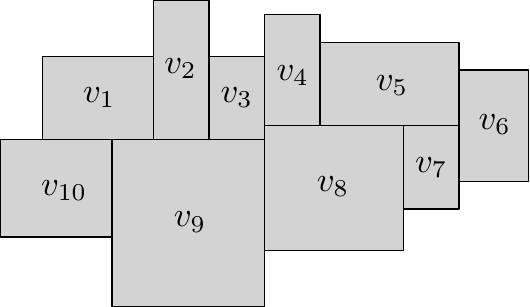}\hspace{1cm}
\includegraphics[width=5cm]{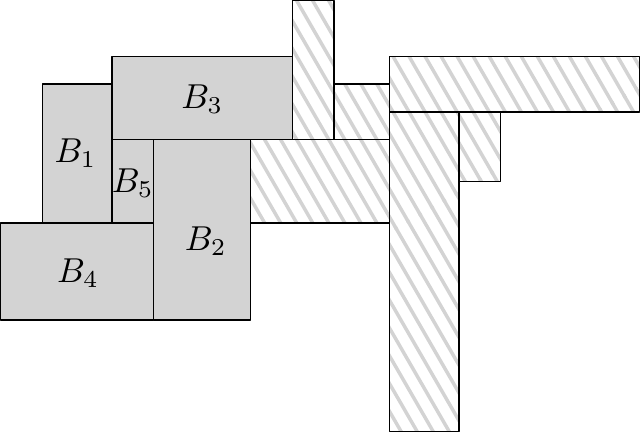}
  \caption{Left: Realizing cycle $(v_1,\ldots,v_{10})$. Right: $8$ adjacencies with $5$ boxes in Lemma~\ref{lem:build-cycle}.}
  \label{fig:combined}
 \end{figure}

 Now consider the box with rightmost right side. If it is $v_i$ we place the box $v_{j-1}$ with its top-left corner onto $p$. If it is $v_j$ we place the box $v_{i+1}$ with its bottom-left corner onto $p$. If the right sides of $v_i$ and $v_j$ are collinear and $j - i > 2$ we place $v_{j-1}$ with its top-left corner slightly above $p$ and $v_{i+1}$ with its bottom-left corner onto the top-left corner of $v_{j-1}$. If $j - i = 2$, that is, $v_{i+1} = v_{j-1}$ is the last box, we place it with its top-left corner slightly above $p$.

 In either case, after each step the current representation realizes a cycle of the form $(v_1,\ldots,v_i,v_j,\ldots,v_n)$ for some $i < j$. In the example in Fig.~\ref{fig:combined} the boxes where added as follows: $\{v_1,v_{10}\}$, $\{v_9\}$, $\{v_2\}$, $\{v_3\}$, $\{v_4,v_8\}$, $\{v_5\}$, $\{v_7\}$, $\{v_6\}$.
\end{proof}

Similar to Theorem~\ref{thm:approx-from-stars}, from Lemma~\ref{lem:build-cycle} we can obtain an approximation algorithm for the \fbcropt{} problem, in case the supporting graph can be covered by few sets of disjoint cycles.

\begin{theorem}\label{thm:approx-from-cycles}
 If one can find in polynomial time $k$ sets of disjoint cycles that together cover the edges of the supporting graph, then one can find in polynomial time a representation with total profit at least $\frac{1}{k} \sum_{i\neq j} p_{ij}.$
 In particular, this is a polynomial-time $1/k$-approximation algorithm for the \fbcropt{} problem.
\end{theorem}

\begin{corollary}\label{cor:delta-approx}
 There is a polynomial-time $\frac{2}{\Delta+1}$-approximation algorithm for the \fbcropt{} problem if the supporting graph has maximum degree $\Delta$.
\end{corollary}
\begin{proof}
 As Peterson shows~\cite{peterson}, the edges of any graph of maximum degree $\Delta$ can be covered by $\lceil\frac{\Delta}{2}\rceil$ sets of cycles, and such sets can be found in polynomial time. The result now follows from Theorem~\ref{thm:approx-from-cycles}.
\end{proof}

\subsection{An Extremal \fbcropt{} Problem}\label{sec:extremal}

Consider a set $\mathcal{B} = \{B_1,\ldots,B_n\}$ of $n$ boxes with fixed dimensions, the complete graph, $G=K_n$, as support graph, and all profits worth 1 unit.
Denote by $f(\mathcal{B})$ the maximum number of adjacencies that can be realized among the $n$ boxes in $\mathcal{B}$. Further we define
$f(n) = \min\{f(\mathcal{B}) \;:\; |\mathcal{B}| = n\}.$

\begin{theorem}\label{thm:extremal}
 For $n = 2,3,4$ we have $f(n) = 2n-3$ and for every $n \geq 5$ we have
 $$f(n) = 2n-2.$$
\end{theorem}
\begin{proof}
 It is easy to verify the lower bound for the base cases $f(n) \geq 2n-3$ for $n=2,3,4$. So let $n \geq 5$ and fix $\mathcal{B} = \{B_1,\ldots,B_n\}$ to be any set of $n$ boxes. We have to show that $f(\mathcal{B}) \geq 2n-2$, i.e., that we can position the boxes so that $2n-2$ pairs of boxes touch. We start by selecting five arbitrary boxes $B_1,B_2,B_3,B_4,B_5$. Without loss of generality, let $B_1$ and $B_2$ be the boxes with largest height, and $B_3$ and $B_4$ be the boxes with largest width among $\{B_3,B_4,B_5\}$. We place the five boxes as in Fig.~\ref{fig:combined}.
 The remaining $n-5$ boxes are added to the picture in any order in such a way that every box realizes two adjacencies at the time it is placed. To this end it is enough to apply the procedure described in Lemma~\ref{lem:build-cycle} taking $B_2,B_3$ as the first two boxes.

Next consider the upper bounds. We have $f(n) \leq 2n-3$ for $n=2,3$ simply because a pair of boxes can touch only once. We have $f(4) \leq 5$ because contact graphs of boxes are planar graphs in which every triangle is an inner face, which rules out $K_4$.
 So let $n \geq 5$. We show that $f(n) \leq 2n-2$, by constructing a set of $n$ boxes for which, in any arrangement of the boxes, at most $2n-2$ pairs of boxes touch. For $i=1,\ldots,n$ we define $B_i$ to be a square box of side length $2^i$. Consider any placement of the boxes $B_1,\ldots,B_n$. We partition the contacts into horizontal contacts and vertical contacts, depending on whether the two boxes touch with horizontal sides or vertical sides. From the side length of boxes, it now follows that neither set of contacts contains a cycle, i.e., consists of at most $n-1$ contacts. This gives at most $2n-2$ contacts in total.
\end{proof}

\section{The \fbcrarea{} problem}\label{sec:area}
Not all contact representations realizing the same adjacencies are equally practically useful (or visually appealing) when viewed as word clouds. Here we consider the \fbcrarea{} problem and show that finding a ``compact'' representation, fitting into a small bounding box, is another hard problem. In particular, we are given a supporting graph $G$, which is known to be realizable
and the goal is to find a representation that still realizes $G$ and additionally fits into a small bounding box.

The reductions are from the (strongly) \NP{}-hard $2$D \prob{Strip Packing} problem, defined as follows. We are given a set $R = \{r_1, r_2, \dots r_n\}$ of $n$ rectangles with height and weight functions: $w: R \rightarrow \N$, $h: R \rightarrow \N$. All the widths and heights are integers bounded by some polynomial in $n$. We are also given a strip of  width $W$ and infinite height and a positive integer $H$, also bounded by a polynomial in $n$. The task is to pack the given rectangles into the strip such that the total height is at most $H$.

The \prob{Strip Packing} problem is actually equivalent to the \fbcrarea{} problem when the supporting graph is an $n$-vertex independent set, because it boils down to deciding whether all the rectangles can be packed into a bounding box of dimensions $W \times H$. However, edges in the supporting graph impose additional constraints on the representation, which might make the \fbcrarea{} problem easier. The following theorem (proof is in the Appendix) shows that this is not the case.

\begin{theorem}
 \fbcrarea{} is \NP{}-hard, even if the supporting graph is a path.
\label{th-path}
\end{theorem}

\section{Experimental Results}\label{sec:experimental}

We implemented the algorithm from Corollary~\ref{cor:approx} for planar graphs (referred to as \textsc{Planar}) and compared it with the algorithm from~\cite{context_preserving_dynamic} (referred to as \textsc{CPDWCV}). Our data set is 120 Wikipedia documents, with 400 words or more. For the word clouds we chose the 100 most frequent words (after removing stop-words, e.g., ``and'', ``the'', ``of''), and constructed supporting graph $G$ with $100$ vertices. Details are provided in the Appendix.

We compare the percentage of realized profit in the representation of $G$
for the two algorithms. Since \textsc{Planar} handles planar supporting graphs, we first extract a maximal planar subgraph $G_{\text{planar}}$ of $G$, and then we apply the algorithm on $G_{\text{planar}}$. For \textsc{CPDWCV} we compute the results for graph $G$. The percentage of realized profit is presented in the table.
Our results indicate that, in terms of the realized profit, \textsc{Planar} performs significantly better than the heuristic \textsc{CPDWCV}. Although we only prove a $\frac{1}{6}\left(1 - \frac{1}{e}\right) \approx 0.1054$-approximation for planar graphs (Corollary~\ref{cor:approx} in combination with Theorem~\ref{thm:approx-star}), in practice \textsc{Planar} realizes more than $25 \%$ of the total profit of planar graphs.

\vspace{10pt}
\begin{center}
  \newcolumntype{Y}{>{\centering\arraybackslash}X}
  \newcolumntype{Z}{>{\raggedleft\arraybackslash}X}
  \begin{tabularx}{0.9\textwidth}{YYl}
    Algorithm & Realized Profit of $G$\:\:\: & Realized Profit of $G_{planar}$ \\
    \toprule
    \textsc{Planar} & $8.56 \%$ & $27.48 \%$ \\
    \textsc{CPDWCV} & $0.77 \%$ &  \\
  \end{tabularx}
\end{center}
\vspace{10pt}


\section{Conclusions and Future Work}\label{sec:conclusions}


We formulated the Word Rectangle Adjacency Contact (\fbcr{}) problem, motivated by the desire to provide theoretical guarantees for semantic-preserving word cloud visualization. We described efficient polynomial-time algorithms for variants of \fbcr{}, showed that some variants are \NP{}-complete, and described several approximation algorithms. A natural open problem is to find an approximation algorithm for general graphs with arbitrary profits.

\medskip\noindent{\bf Acknowledgements:} Work on this problem began at Dagstuhl Seminar 12261.  We thank the organizers, participants, and especially Steve Chaplick, Sara Fabrikant, Anna Lubiw, Martin N\"ollenburg, Yoshio Okamoto, G\"unter Rote, Alexander Wolff.

\bibliographystyle{abbrv}
{\small
\bibliography{literature,refs}
}

\newpage\section*{Appendix}


\begin{proof}[Proof of Theorem~\ref{thm:quasi-triangulated}]
 Let $G$ be the supporting, quasi-triangulated graph. We consider $G$ embedded in the plane  with outer face $\{v_N,v_E,v_S,v_W\}$. Note that this embedding is unique. Abusing notation, we refer to a vertex and its corresponding box with the same letter.

 We begin by placing a horizontal and a vertical ray emerging from the same point in positive $x$-direction and positive $y$-direction, respectively. For the first phase of the algorithm let us pretend that the horizontal ray is the box $v_S$ (imagine a rectangle with tiny height and huge width) and the vertical ray is the box $v_W$ (imagine a rectangle with tiny width and huge height), independent of how the actual boxes look like; see Fig.~\ref{fig:quasi-triangulated_1}.

 \begin{figure}[h]
  \centering
  \subfloat{\includegraphics{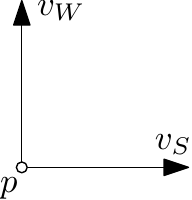}
  }\hspace{2em}
  \subfloat{\includegraphics{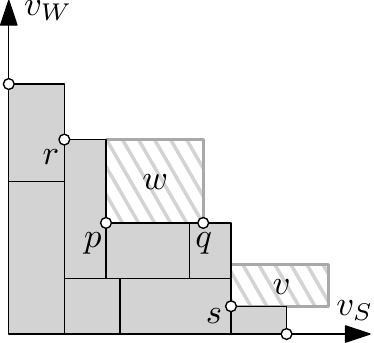}
  }
  \hspace{2em}
  \subfloat{\includegraphics{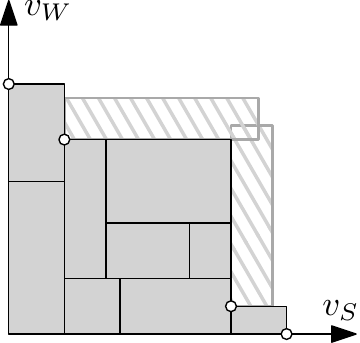}
  }
  \caption{Left: starting configuration with rays $v_S$ and $v_W$.
Center: representation at an intermediate step: vertex $w$ fits into concavity $p$ and is applicable, vertex $v$ fits into concavity $s$ but is not applicable. Adding box $w$ to the representation introduces new concavity $q$, and the vertex at concavity $r$ may become applicable. Right: there is no applicable vertex and the algorithm terminates.}
  \label{fig:quasi-triangulated_1}
 \end{figure}

 We build up a representation by adding one rectangle at a time. At every intermediate step the representation is \emph{rectilinear convex}, that is, its intersection with any horizontal or vertical line is connected. In other words, the representation has no holes and a ``staircase shape''. We maintain the set of all \emph{concavities}, that is, points on the boundary of the representation, which are bottom-right or top-left corners of some rectangle but not a top-right corner of any rectangle. Initially there is only one concavity, namely the point where the rays $v_W$ and $v_S$ meet.

 Each concavity $p$ is a point on the boundary of two rectangles, say $u$ and $v$. Since $G$ has no separating triangles there are exactly two vertices that are adjacent to both, $u$ and $v$, or only one if $\{u,v\} = \{v_S,v_W\}$. For exactly one of the these vertices, call it $w$, the rectangle is not yet placed because its bottom-left corner is supposed to be placed on the concavity $p$. We say that $w$ \emph{fits into the concavity $p$}. We call a vertex $w$ \emph{applicable} to an intermediate representation if it fits into some concavity and adding the rectangle $w$ gives a representation that is rectilinear convex. In the very beginning the unique common neighbor of $v_S$ and $v_W$ is applicable.

 The algorithm proceeds in $n-4$ steps as follows. At each step we identify a inner vertex $w$ of $G$ that is applicable to the current representation. We add the rectangle $w$ to the representation and update the set of concavities and applicable vertices. At most two points have to be added to the set of concavities, while one is removed from this set. The vertices that fit into the new concavities can easily be read off from the plane embedding of $G$. Checking whether these vertices are applicable is easy. If the top-left or bottom-right corner of $w$ does not define a concavity then one has to check whether the vertices that fit into existing concavities to the left or below, respectively, are now applicable. So each step can be done in constant time.

 If the algorithm has placed the last inner vertex, it suffices to check whether the representation without the two rays is a rectangle, that is, whether there are exactly two concavities left. If so, call this rectangle $R$, we check whether the width of $R$ is at most the width of $v_N$ and $v_S$ and whether the height of $R$ is at most the height of $v_E$ and $v_W$. If this holds true, we can easily place the rectangles $v_N$, $v_E$, $v_S$, $v_W$ to get a representation that realizes $G$. The total running time is linear.

 On the other hand, if the algorithm stops because there is no applicable vertex, or the height/width-conditions in the end phase are not met, then there is no representation that realizes $G$. This is due to the lack of choice in building the representation -- if a vertex $v$ is applicable to a concavity $p$ then the bottom-left corner of $v$ has to be placed at $p$ in order to establish the contacts of $v$ with the two rectangles containing $p$.
\end{proof}


\begin{proof}[Proof of Theorem~\ref{thm:star-hardness}]
 We use a reduction from \prob{Knapsack}, which is defined as follows. Given a set of $n$ items, each with a positive weight $w_i$, $i=1,\ldots,n$, a positive profit $p_i$, $i=1,\ldots,n$, a knapsack with some positive capacity $C$, and a positive number $P$, the task is to find a subset of items whose sum of weights does not exceed $C$ and whose sum of profits is at least $P$. This classical problem is known to be weakly \NP{}-complete.

 The reduction is similar to the one presented in the proof of Theorem~\ref{thm:trees:hardness}. Given an instance $I = \{(w_1,p_1),\ldots,(w_n,p_n),C,P\}$ of \prob{Knapsack} we define an edge-weighted star $S_I$ on $n+5$ vertices as follows. There is a vertex $v_i$ for each $i=1,\ldots,n$, a vertex $c$, and five vertices $a_1,a_2,a_3,a_4,a_5$. Vertex $c$ is the center of the star $S_I$, its edge to $v_i$ has weight $p_i$ for $i=1,\ldots,n$, and its edge to $a_k$ has weight $\sum_{i=1}^n p_i$ for $k = 1,2,3,4,5$; see Fig.~\ref{fig:star-hardness}.

 \begin{figure}[h]
  \centering
  \subfloat{\includegraphics[width=.25\textwidth]{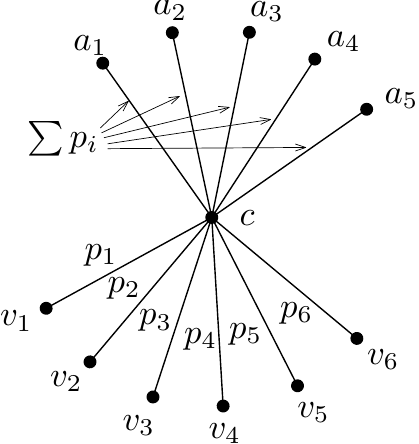}
  }\hspace{3em}
  \subfloat{\includegraphics{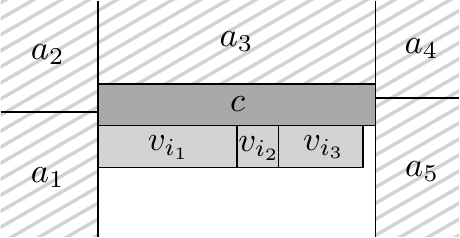}
  }
  \caption{Left: edge-weighted star $S_I$, defined from instance $I$ of the \prob{Knapsack} problem. Right: optimal solution to\fbcropt{} for $S_I$.}
  \label{fig:star-hardness}
 \end{figure}

As before, we use $v \to (h,w)$ to define the box of $v$ with height $h$ and width $w$. We define $v_i \to (1,w_i)$ for $i=1,\ldots,n$, $a_k \to (C,C)$ for $k=1,2,3,4,5$, and $c \to (1,C)$. Finally, we define the target profit in the \fbcropt{} problem  $P_I=5\sum_{i=1}^n p_i + P$.

 We claim that an instance $I$ of the \prob{Knapsack} problem is feasible if and only if the instance of the \fbcropt{} problem corresponding to $S_I$ is feasible. From any solution of the \fbcropt{} problem we can read off a solution for the \prob{Knapsack} problem.

 First note that every solution of the \fbcropt{} problem has total profit strictly more than $5\sum_{i=1}^n p_i$. Thus all adjacencies between $c$ and $a_k$ for $k=1,2,3,4,5$ are realized and each $a_k$ contains a corner of $c$. It follows that at least three sides of $c$ are partially covered by some $a_k$ and at least one horizontal side of $c$ is completely covered by some $a_k$. Because $c$ has height $1$ none of the boxes $v_1,\ldots,v_n$ (each of height $1$) touches $c$ on the side. Hence each $v_i$ touches $c$ (if at all) on a horizontal side, say the bottom; see Fig.~\ref{fig:star-hardness}.

 Now the bottom side of $c$ has width $C$ and each box $v_i$ has width $w_i$, $i=1,\ldots,n$. Thus the subset of $J \subseteq \{1,\ldots,n\}$ of indices of boxes that touch $c$ satisfies $\sum_{j \in J} w_j \leq C$. Moreover the total profit of the representation is $5\sum_{i=1}^n p_i + \sum_{j \in J} p_j$, which is at least $P_I$ if and only if $\sum_{j \in J} p_j \geq P$, that is, the items with indices in $J$ are a solution of the \prob{Knapsack} problem.

 Along the same lines, we can construct a solution for the \fbcropt{} problem based on any solution of the \prob{Knapsack} problem, and this concludes the proof.
\end{proof}

\begin{proof}
 We use a reduction from \prob{Strip Packing}, so fix any instance $I$ of \prob{Strip Packing} consisting of rectangles $r_1,\ldots,r_n$ and two integers $H$ and $W$. Let $d = \frac{\epsilon}{\max(W,H)}$ for some $\epsilon \in (0,1)$.

 We define an instance of the \fbcrarea{} problem by slightly increasing the heights and widths in $I$. The idea is to lay a unit square grid over the strip and blow each grid line up to have a thickness of $d$; see Fig.~\ref{fig:paths-hard-grid}. Each rectangle in $I$ is stretched according to the number of grid lines is intersects.

 \begin{figure}[h]
  \centering
  \subfloat{\includegraphics[width=0.4\textwidth]{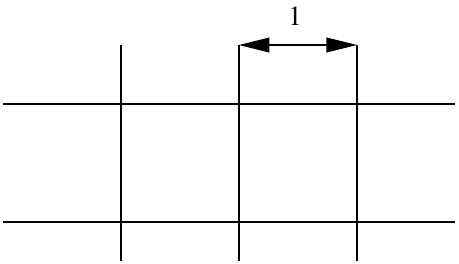}}
  \hspace{0.05\textwidth}
  \subfloat{\includegraphics[width=0.4\textwidth]{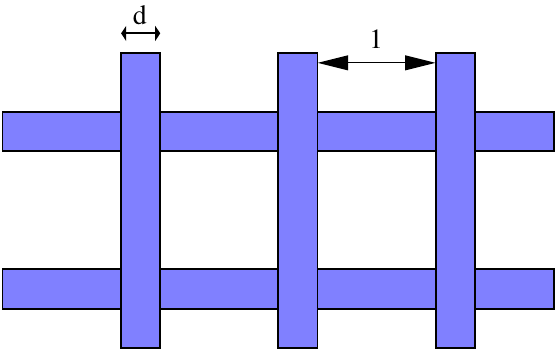}}
  \caption{Grid before and after stretching}
  \label{fig:paths-hard-grid}
 \end{figure}

 More precisely, we define for $i=1,\ldots,n$ a rectangle $r'_i$ of width $w(r_i) + (w(r_i) - 1) d$ and height $h(r_i) + (h(r_i) - 1) d$. Further we define $W' = W + (W-1)d$ and $H' = H + (H-1)d$. Finally, we arrange the rectangles $r'_1,\ldots,r'_n$ into a path $P$ by introducing between $r_i$ and $r_{i+1}$ ($i=1,\ldots,n-1$), as well as before $r'_1$ $k$ small $x \times x$ square, called \emph{connector squares}. We choose $k$ and $x$ to satisfy

 \begin{align}
  kx &= 4(n+3)(H + 2nW) \hspace{4em}\text{and}\label{eqn:connector-length}\\
  n(kx^2 + 2x) &= d.    \label{eqn:connector-space}
 \end{align}

 In particular, we choose
 \begin{align*}
  x &= \frac{d}{2n(2Hn+6H+4n^2W+12nW+1)} \hspace{4em}\text{and}\\
  k &= \frac{4(n+3)(H+2nW)}{x}.
 \end{align*}

 We claim that there is a representation realizing $P$ within the $W' \times H'$ bounding box if and only if the original rectangles $r_1,\ldots,r_n$ can be packed into the original $W \times H$ bounding box.

 First consider any representation realizing $P$ within the $W' \times H'$ bounding box and remove all connector squares from it. Since $W' < W + \epsilon < W + 1$ and $H' < H + \epsilon < H + 1$, the stretched bounding box has the same number of grid lines than the original. Hence the rectangles $r'_1,\ldots,r'_n$ can be replaced by the corresponding rectangles $r_1,\ldots,r_n$ and perturbed slightly such that every corner lies on a grid point. This way we obtain a solution for the original instance of \prob{Strip Packing}.


 Now consider any solution for the \prob{Strip Packing} instance, i.e., any packing of the rectangles $r_1,\ldots,r_n$ within the $W \times H$ bounding box. We will construct a representation realizing the path $P$ within the $W' \times H'$ bounding box. We start blowing up the grid lines of the $W \times H$ bounding box to thickness $d$ each, which also effects all rectangles intersected by a grid line in its interior. This way we obtain a placement of bigger rectangles $r'_1,\ldots,r'_n$ $I'$ in the bigger $W' \times H'$ bounding box, such that every rectangle $r'_i$ intersects the interiors of exactly those blown-up grid lines corresponding to the grid lines that intersect $r_i$ interiorly. Thus any two rectangles $r'_i$ and $r'_j$ are separated by a vertical or horizontal corridor of thickness at least $d$. We will refer to the grid lines of thickness $d$ as \emph{gaps}.

 It remains to place all the connector square so as to realize the path $P$. The idea is the following. We start in the lower left corner of the bounding box, and lay out connector squares horizontally to the right inside the bottommost horizontal gap until we reach the vertical gap that contains the lower-left corner of $r'_1$. We then start laying out the connector squares inside this vertical gap upwards, until we reach the lower-left corner of $r'_1$. Whenever a rectangle $r'_i$ overlaps with this vertical gap, we go around $r'_i$ as illustrated in Fig.~\ref{fig:paths-conn-rerouting-after}. This way we lay out at most $(3W'+H')/x$ connector squares, which by~\eqref{eqn:connector-length} is less than $k$. The remaining connector squares are ``folded up'' inside the vertical gap; see Fig.~\ref{fig:hardness-path-folding}.

%
%
 \begin{figure}[h!]
  \begin{tabular}{cc}
   \subfloat[\NP{} hardness proof for \fbcrarea{} of paths]{\label{fig:hardness-path-packing} \includegraphics[width=.5\textwidth]{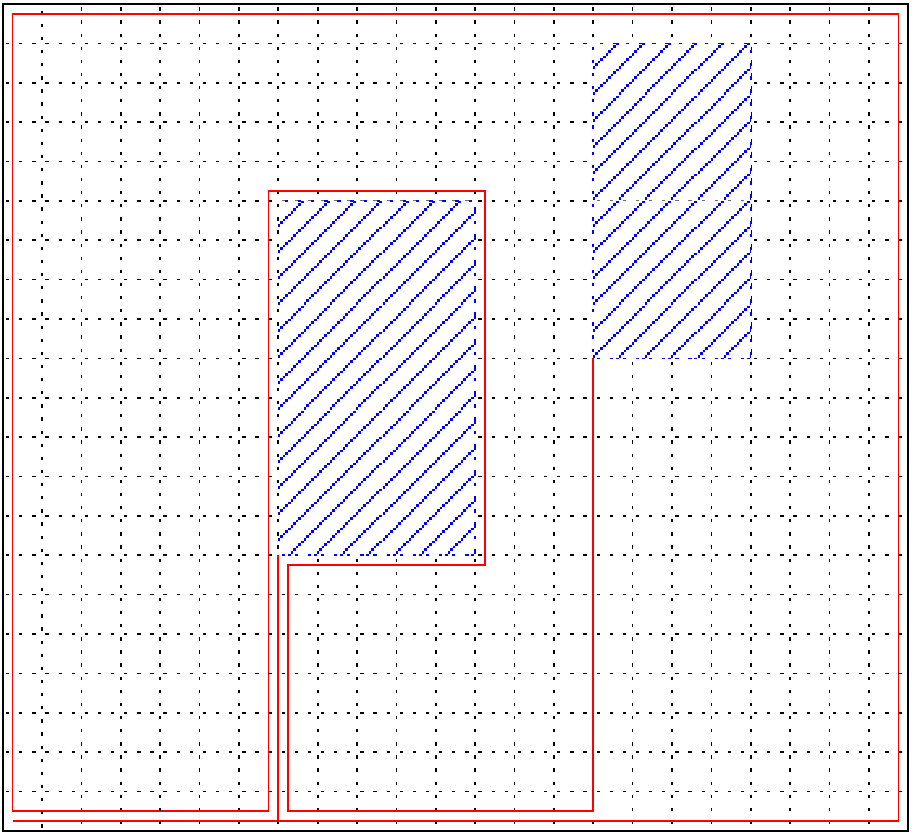}}
   & \subfloat[Folding of connector rectangles inside a gap]{\label{fig:hardness-path-folding}\includegraphics[height=.4\textheight]{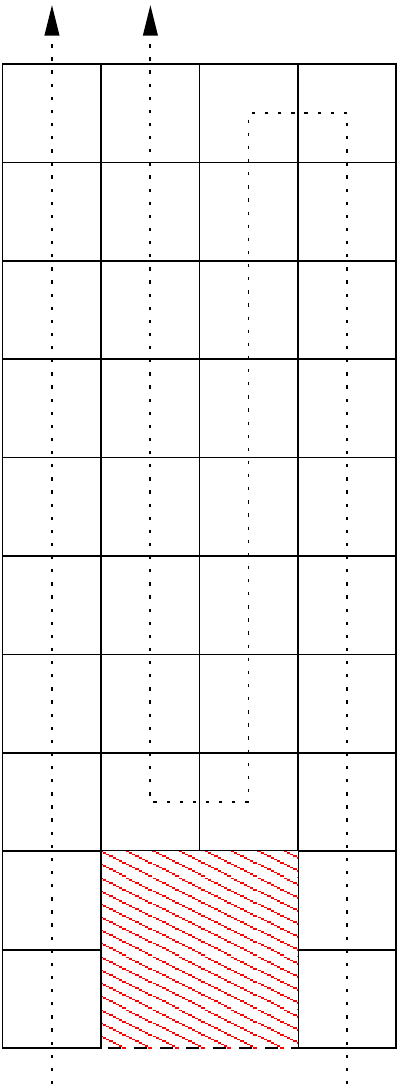}}\\
   \subfloat[Connectors before rerouting]{\label{fig:paths-conn-rerouting-before}\includegraphics[height=.35\textheight]{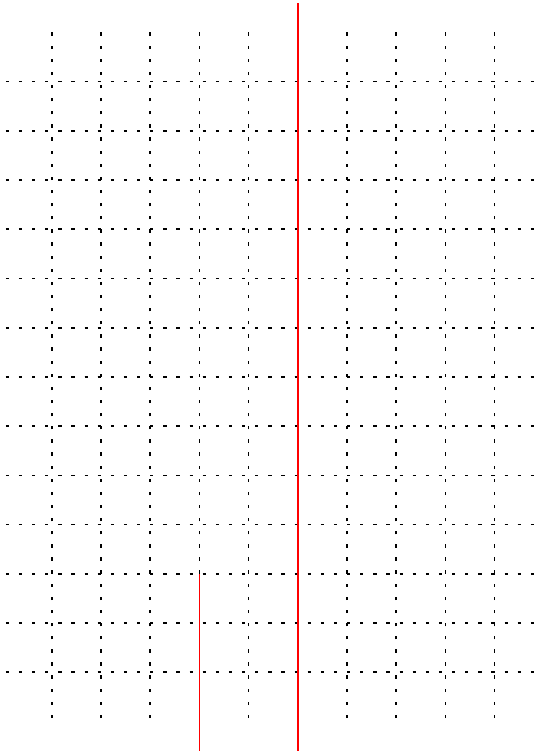}}
   & \subfloat[Connectors after rerouting]{\label{fig:paths-conn-rerouting-after}\includegraphics[height=.35\textheight]{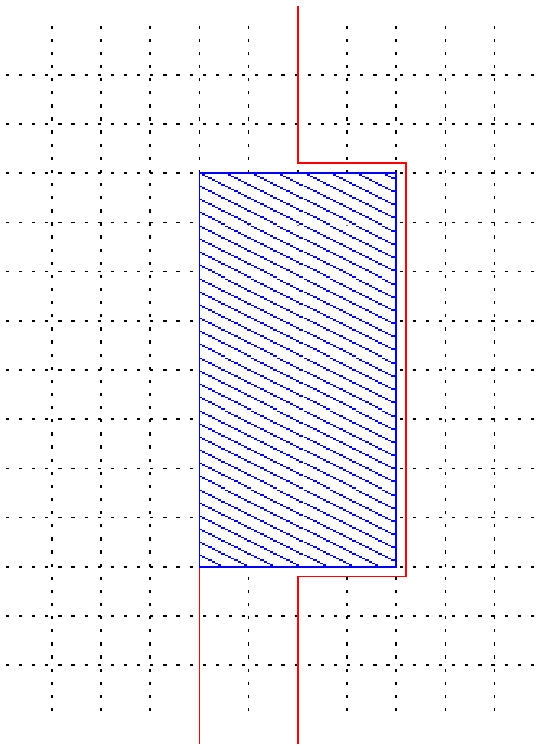}}
  \end{tabular}
\caption{Illustrations for Theorem~\ref{th-path}.}
 \end{figure}

 Next we lay out the connectors squares between $r'_1$ and $r'_2$. We start where we ended before, i.e., at the lower-left corner of $r'_1$, and go the along the path we took before till we reach the bottommost gap. Then we lay connector squares along the outermost gaps in counterclockwise direction, i.e., first horizontally to the rightmost gap, then up to the topmost gap, left to the leftmost gap, and down to the bottommost gap. Now we do the same for $r'_2$ than what we did for $r'_1$. If while going right we ``hit'' the connector squares going up to $r'_1$, we follow them up, go around $r'_1$, and go down again. This is possible since there are gaps all around $r'_1$; see Fig.~\ref{fig:hardness-path-packing}. Note that the red line of connectors will actually sit \textit{on} the dashed, expanded grid lines but are drawn next to them for better readability.


 We repeat this for all the rectangles.

 We have to show two things: The number of connector squares between two $r'_i$ and $r'_{i+1}$ is large enough so that the length of the string of connectors is sufficient. And that the gaps have sufficient space so that we can fold up the connectors in them.

 The first condition is taken care of by equation~\eqref{eqn:connector-length}. We divide the path of the connectors in up to $n + 3$ parts: The first part $p_{{\text{down}}_i}$ is going down from $r'_i$ to the bottom gap. The second part $p_{\text{circle}}$ that goes around the bounding box in counterclockwise order to the vertical gap containing the lower-left corner of $r'_{i+1}$. This part is intercepted by up to $n$ parts $p_{{\text{avoid}}_k}$ where we hit a string of connectors going up to another rectangle $r'_k$ and we have to follow it, go around $r'_k$ and come down again. The last part $p_{{\text{up}}_{i+1}}$ is going up from the bottom gap to the position of $r'_{i+1}$. We will now show that each of these parts has a maximum length of $4(H' + 2nW')$.

 The parts $p_{{\text{up}}_{i+1}}$ and $p_{{\text{down}}_i}$ have to span the height $H'$ at most once, and may encounter all other rectangles $r'_k$ at most once. Going around any such $r'_k$ means at most traversing its width twice, which is at most $2W'$. Hence each of $p_{{\text{up}}_{i+1}}$ and $p_{{\text{down}}_i}$ has a total length of at most $H' + 2nW' < 4(H' + 2nW')$. Since every $p_{{\text{avoid}}_k}$ exactly follows the $p_{{\text{up}}_k}$, then surrounds $r'_k$ (which has maximum width $W'$ and maximum height $H'$) and then follows $p_{{\text{down}}_k}$, it has a maximum length of $2(W'+H') + 2(H' + 2nW') \leq 4(H' + 2nW')$. Finally, $p_{\text{circle}}$ has a maximum length of $2H' + 2W' \leq 4(H' + 2nW')$.

 Thus, the total length of the path of connectors comprised of $n + 3$ parts of at most length $4(H' + 2nW')$ each is at most $4(n+3)(H' + 2nW')$. Equation~\eqref{eqn:connector-length} ensures that our string of connectors has sufficient length.

 The second condition is covered by equation~\eqref{eqn:connector-space}. Consider Fig.~\ref{fig:hardness-path-folding}. If a string of connectors just passes through a gap, it takes up exactly $1 \times x$ space. If it folds $m$ connector rectangles inside the gap, it takes $m \times x^2$ plus the 'wasted' space (the red shaded space in Fig.~\ref{fig:hardness-path-folding}). The wasted space can be at most $1 \times 2x$, and since every string of connectors has $k$ connector rectangles, the space taken up by those can be at most $kx^2$, thus every string of connectors can take at most $kx^2 + 2x$ space in any given gap. Since there are $n$ such strings of connectors and every gap has dimensions $1 \times d$, equation~\eqref{eqn:connector-space} ensures that the space in every gap is sufficient.



We showed that we can find a layout of the path that corresponds to the optimum packing of the rectangles, if such a packing exists within the desired bounding box. Thus, finding the most space-efficient layout for a path of rectangles is \NP-hard.
\end{proof}

\paragraph{Implementation Details} Here we provide some details regarding the implementation of the algorithms \textsc{Planar} and \textsc{CPDWCV} from Section~\ref{sec:experimental}.

Before the algorithms are applied, the text is preprocessed using this workflow: The text is split into sentences, and the sentences are split into words using Apache OpenNLP. We then remove stop words, perform stemming on the words and group the words with the same stem. The similarity of words is computed using Latent Semantic Analysis based on the co-occurrence of the words within the same sentence.

In the implementation of \textsc{Planar}, we use the $(\frac{\beta}{\beta + 1} - \epsilon)$-approximation from~\cite{Fleischer2011} combined with a FPTAS for \textsc{Knapsack} to approximate the stars.
In the implementation of \textsc{CPDWCV}, we achieved the best results in our experiments with 
parameters $K_r = 4000$ and $K_a = 25$. One of the results computed by our algorithm is given in Fig.~\ref{fig:wordle}.

\begin{figure}[t]
 \centering
 \includegraphics[width=.7\textwidth]{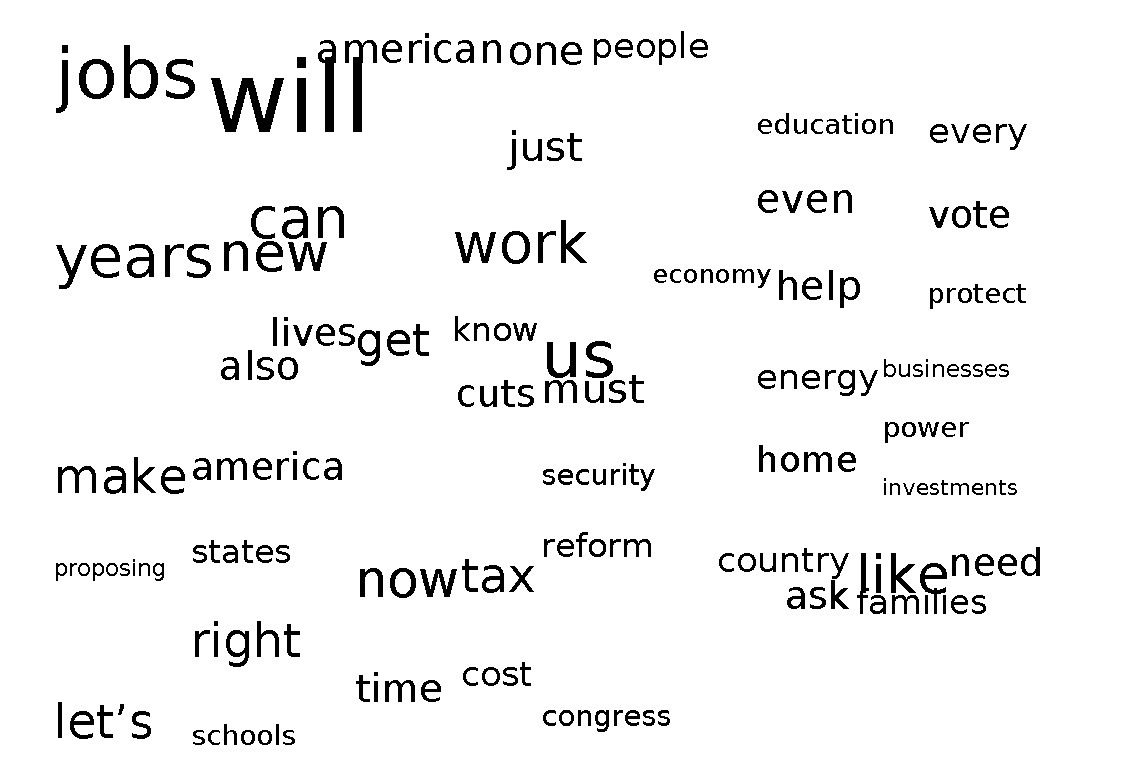}
 \caption{A result of the \textsc{Planar} algorithm: Star-Based semantic preserving visualization of Obama's 2013
State of the Union Speech.}
 \label{fig:wordle}
\end{figure}

\end{document}

%% file: macros.tex
\newcommand{\prob}[1]{\text{\textsc{#1}}}
\newcommand{\blink}[1]{\textnormal{\texttt{#1}}}

\newcommand{\NP}[0]{\blink{NP}}
\newcommand{\NPI}[0]{\blink{NPI}}
\newcommand{\NPC}[0]{\blink{NPC}}
\newcommand{\coNP}[0]{\ensuremath{\text{co-}\mathcal{NP}}}
\newcommand{\coNPC}[0]{\ensuremath{\text{co-}\mathcal{NPC}}}
\renewcommand{\P}[0]{\blink{P}}
\newcommand{\BPP}[0]{\blink{BPP}}
\newcommand{\SimAn}[0]{\blink{Simulated Annealing}}
\newcommand{\Prob}[0]{\ensuremath{\mathbb{P}} }

\newtheorem{lem}{Lemma}

\newcommand{\R}{{\ensuremath{\mathbb{R}}}}
\newcommand{\N}{{\ensuremath{\mathbb{N}}}}
\newcommand{\Z}{{\ensuremath{\mathbb{Z}}}}
\newcommand{\C}{{\ensuremath{\mathbb{C}}}}
\newcommand{\Q}{{\ensuremath{\mathbb{Q}}}}
\newcommand{\F}{{\ensuremath{\mathbb{F}}}}
\newcommand{\Prim}{{\ensuremath{\mathbb{P}}}}
\newcommand{\X}{{\ensuremath{\mathbb{X}}}}

\newcommand{\Oh}{{\ensuremath{\mathcal{O}}}}


\renewcommand{\phi}{\varphi}
\renewcommand{\epsilon}{\varepsilon}

\newcommand{\fbcr}{\textsc{WRAC}} 
\newcommand{\fbcropt}{\textsc{Max}-\textsc{WRAC}} 
\newcommand{\fbcrhier}{\textsc{Hi}-\textsc{WRAC}} 
\newcommand{\fbcrarea}{\textsc{Area}-\textsc{WRAC}} 

\newcommand{\ggT}{\mathop{\mathrm{ggT}}}

\makeatletter
\def\imod#1{\allowbreak\mkern10mu({\operator@font mod}\,\,#1)}
\makeatother